\documentclass[11pt]{article}
\usepackage{fullpage}
\usepackage[utf8]{inputenc}
\usepackage[english]{babel}
\usepackage[T1]{fontenc}
\usepackage{amsmath,amssymb,amsfonts,amsthm}
\usepackage{url}
\usepackage{hyperref}
\usepackage{cleveref}
\usepackage{graphicx}
\usepackage{makecell}
\usepackage{algorithm}
\usepackage{algpseudocode}
\usepackage{tikz}
\usepackage{physics}
\usepackage{mathcommands}

\usepackage[backend=bibtex,doi=false,isbn=false,url=false,maxbibnames=5,sorting=none]{biblatex}
\AtEveryBibitem{\clearfield{month}}
\AtEveryBibitem{\clearfield{day}}
\newtheorem{theorem}{Theorem}
\newtheorem*{theorem*}{Theorem}
\newtheorem{lemma}[theorem]{Lemma}

\newtheorem{definition}[theorem]{Definition}

\usepackage{authblk}

\def\cR{\mathcal{R}}

\def\fB{\mathfrak{B}}

\usepackage{xcolor}

\title{Stabilizer codes of less than two dimensions have constant distance}

\author[1]{Nou\'edyn Baspin}

\affil[1]{Centre for Engineered Quantum Systems, School of Physics, University of Sydney, Sydney, NSW 2006, Australia}

\begin{document}
\maketitle

\begin{abstract}
   The surface code is a two-dimensional stabiliser code with parameters $[[n,1,\Theta(\sqrt{n})]]$. To this day, no stabiliser code with growing distance is know to live in less than two dimensions. In this note we show that no such code can exist. 
\end{abstract}


\section{Introduction}

The information contained in a bit corresponds to a simple bit-flip: the bit is either up, or down. In a repetition code -- an elementary classical code -- this information is encoded under the form of a long, one-dimensional, `string-like' logical operator.

In contrast, a qubit holds two pieces of information: a bit-flip, and a phase-flip. As might then be expected, the surface code \cite{bravyi1998quantum} -- a quantum equivalent of the repetition code -- exhibits two string-like operators, one for each of bit-flips and phase-flips. 
Notably, each of these string operators occupies a different dimension, which `forces' the surface code to occupy two dimensions.

The correlation between dimensionality and number of logical operators is striking. One could hope that a smarter encoding could improve over using orthogonal string operators, and a code of lower dimension could be invented. Unfortunately, in the 25 years that have elapsed since the conception of the surface code, no code family of dimension less than two has been discovered. 

In this work we resolve this question by proving that quantum stabiliser codes need at least two dimensions for their distance to grow. 

\begin{theorem}[Informal statement of Theorem \ref{thm:main-thm}]
    Let $\cC$ be a stabilizer code with connectivity graph $G$. If $G$ has Assouad dimension at most $\beta<2$, then the distance $d$ of $\cC$ obeys $d \leq O(1)$.
\end{theorem}

The interaction between locality and quantum error correction has been feeding a sustained interest \cite{portnoy2023local, pattison2023hierarchical, hong2024long, dai2024locality, baspin2021quantifying}. The impetus for this line of research can be summarised as follows: given a quantum computer architecture, how efficiently can it be protected from errors? In particular one hopes to protect as much information as possible from as many errors as doable.

This study was initiated by \cite{bravyi2009no}, and quickly followed by \cite{bravyi2010tradeoffs}, which showed that the geometry of an architecture strongly impacted the efficiency of the codes it could accommodate. Specifically, they obtained bounds on the code's parameters in function of the dimension it lives in. The motivation behind these original papers was to show that the surface code corresponded to the `optimal' 2D code. Our result can be understood as an extension of these original observations: a quantum code with growing distance takes at least as much space as a surface code.

\section{Codes and their graphs}

We begin our preliminaries with a short review of quantum codes. An $n$-qubit quantum code is a subspace $\cC$ of $(\mathbb{C}^2)^n$, and this code is said to encode $k$ qubits if $\cC \cong (\mathbb{C}^2)^k$. 
A subset $E \subset [n]$ of qubits is called {correctable} if there exists a CPTP map $\cR_E$ such that for any state $\rho$ in the support of $\cC$, we have $\cR_E \circ \tr_E (\rho) = \rho$.
The distance $d$ of $\cC$ corresponds to the cardinality of the smallest non-correctable set of qubits. 

We move on to introducing stabilizer codes, as they will constitute our main object of study. Remember that a stabilizer group is an abelian subgroup $\cS$ of the $n$-qubit Pauli matrices $\cP_n = \{\idty, X, Y, Z\}^{\otimes n}$. A code $\cC$ is then said to be a stabilizer code if:
\[
\cC = \{\ket{\psi}, s \ket{\psi} = \ket{\psi}, \forall s \in \cS\}
\]

The stabiliser group makes for a compact description of $\cC$: stabilizer codes are often identified with a set $\cG \subset \cS$ that generates $\cS$. 
This set also allows us to assign a locality structure to $\cC$, through the connectivity graph.

\begin{definition}[Connectivity graph]
	\label{def:connectivity}
	Let $\cC$ be a stabilizer code on $n$ qubits with generating set $\cG$; then the connectivity graph $G = (V, E)$ associated with these stabilisers is defined as:
	\begin{enumerate}
		\item $V = [n]$, i.e. each vertex is associated with a qubit, and
		\item $(u,v) \in E$ if and only if there exists a generator $g \in \cG$ such that $u,v \in \supp(g)$.
	\end{enumerate}
    Further we denote $\bdry U \subset E$ the set of edges with exactly one endpoint in $U$.
\end{definition}

The seminary results of Bravyi, Poulin, and Terhal (BPT) \cite{bravyi2009no, bravyi2010tradeoffs} on $D$-dimensional lattices provided a framework to describe how the locality structure of the connectivity graph can be used to extract upper bounds on the parameters $k$ and $d$ of the associated code. The graph theoretic language was introduced in \cite{baspin2021connectivity} and was expanded on by later results \cite{baspin2021quantifying, baspin2023combinatorial}.

The Union Lemma is a workhorse of the aforementioned framework, and will play an equally important role in our discussion. It states that correctable subsets that are spatially separated are jointly correctable.

\begin{lemma}[Union Lemma, see \cite{bravyi2010tradeoffs} \footnote{We use the formulation of Lemma 5 in \cite{baspin2021connectivity}.}]
\label{lem:union-lemma}
    Let $G = (V,E)$ be the connectivity graph of a code $\cC$. Let $\{U_i\}_i$ be a collection of subsets of $V$. If every $U_i$ is correctable, and the subsets satisfy
    \[
    \forall i, \forall j \neq i, \partial U_i \cap U_j = \emptyset
    \]
    Then the set $T = \bigcup_i U_i $ is correctable too.
\end{lemma}

\section{Assouad and Nagata dimensions}

In the present discussion we are particularly interested in the `dimensionality' of the connectivity graph $G$. We want to investigate graphs of less than two dimensions, but more than one \footnote{As a reference point, the BPT bound \cite{bravyi2010tradeoffs} reads $kd^{2/(D-1)} \lesssim n$, where $D$ the dimension is an integer. Even if one could generalise it to fractional values of $D$, plugging in $1<D<2$ does not seem to forbid the existence of codes in that range. This observation emphasises the need to go beyond simply generalising the BPT bound. Similarly for the Bravyi-Terhal bound \cite{bravyi2009no}.}. Our first concern is thus to obtain a definition of dimension that is well behaved for fractional values.
Many generalisations of `dimensions' have been studied (Hausdorff, Lebesgue, etc.), however only a few are compatible with discrete metric spaces such as graphs. 

This scarcity will allow us to concentrate on two such notions. The first, the Nagata dimension\footnote{Also frequently called the Assouad-Nagata dimension.}, is an integer number and interacts very conveniently with the Union Lemma. The second, the Assouad dimension, is a real number and allows us to formally and intuitively define what constitutes a graph of `less than two dimensions'. Crucially, the Assouad dimension upper bounds the Nagata dimension.
We wield these two definitions to achieve different purposes. First, we show that a graph of Nagata dimension $1$ implies a code of bounded distance. Secondly, we leverage the fact that graphs of Assouad dimension strictly less than $2$ have Nagata dimension $1$, or less. We conclude that codes whose connectivity graphs have (Assouad) dimension strictly less than $2$ have bounded distance.

Before formally introducing these notions, we establish some elements of notation. The \emph{ball} of radius $r$ with centre $v \in V$ is the set containing all the vertices of $G$ at a distance less than or equal to $r$ of $v$. Two subsets $U', V' \subset V$ are \emph{$s$-separated} if all the vertices of $U'$ are at a distance at least $r$ of all the vertices of $B$. A collection of subsets is called $r$-separated if each distinct pair of element in it is $r$-separated.
A \emph{cover} of a subset of vertices $U \subset V$ is an ensemble of subsets $\{B_j\}_j$ such that $ U \subset \cup_j B_j$. The ensemble is called \emph{$D$-bounded} if the vertices of $B_j$ are at a distance at most $D$ of each other, for every $j$.

\begin{definition}[\cite{ledonne2015assouad}]
\label{def:assouad-dim}
    A graph $G = (V,E)$ is said to have Assouad dimension at most $\beta$ at scale $R$ with constant $C$, if for all $0<r<R$, any ball of radius $R$ in $G$ can be covered with at most $C(R/r)^\beta$ balls of radius $r$. 
\end{definition}

\begin{figure}[ht!]
		\centering
		\includegraphics[page=1, scale=.8]{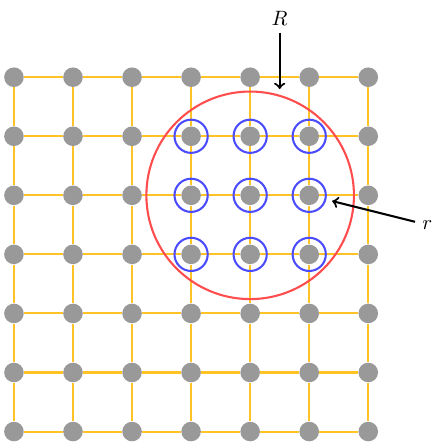}
		\caption{Assouad dimension of a 2D lattice. Here the red circle corresponds to a ball of radius $R$, while the blue circles have radius $r$. In this case, any ball of radius $R$ can be covered by at most $\propto (R/r)^2$ balls of radius $r$. Intuitively, the Assouad dimension is intimately related to the notion of growth of a graph.}
		\label{fig:assouad-dim}
	\end{figure}

Due to its intuitive nature -- see Figure \ref{fig:assouad-dim} -- the Assouad dimension makes for a very \emph{practical} metric. Any graph of `dimension less than two' should also grow slower than a 2D lattice, and therefore have Assouad dimension less than two. Having motivated the use of our first metric, we now move on to introducing the Nagata dimension.

\begin{definition}[Using the formulation of \cite{distel2023proper}]
    A graph $G = (V,E)$ is said to have Nagata dimension at most $m$ at scale $r$ with constant $c>0$ if there exist $m+1$ collections of subsets $\fB_1, \dots, \fB_{m+1}$ such that
    \begin{enumerate}
        \item $\bigcup_{i=1}^{m+1} \bigcup_{B \in \fB_i} = V$
        \item $\forall i,$ $\fB_i$ is $r$-separated
        \item $\forall i, \forall B \in \fB_i$, $B$ is $c r$-bounded 
    \end{enumerate}
\end{definition}

\begin{theorem}[Theorem 1.1 of \cite{ledonne2015assouad}]
\label{thm:assouad-to-nagata}
    Let $r_*, \beta_*, C_*$ be arbitrary real numbers, then there exists a constant $c_*>0$, and a real number $R_*$ such that the following holds:
    
    Any graph $G$ with Assouad dimension at most $\beta_*$ at scale $R_*$ with constant $C_*$, also has Nagata dimension at most $\lfloor\beta_*\rfloor$ at scale $r_*$ with constant $c_*$. 
\end{theorem}

Theorem \ref{thm:assouad-to-nagata} essentially serves as a levelling mechanism: it allows us to treat a graph with arbitrary real-valued dimension $\alpha$ as a graph of `flatter' dimension $\lfloor \alpha \rfloor$, by changing the operational meaning of `dimension'. As a matter of curiosity, one can verify that the converse of the above theorem does not hold: a tree graph has Nagata dimension one, yet its growth is exponential.

\section{Less-than-2D codes have constant distance}

In this section we prove our main result. We begin by reminding the reader of a keystone lemma in BPT. It allows us to bound $k$ as function of partitioning $G$ into correctable regions.
\begin{lemma}[BPT Lemma, Eq. 14 of \cite{bravyi2010tradeoffs} \footnote{See Lemma 19 of \cite{baspin2021connectivity} for a bound explicitly for stabilizer codes}]
\label{lem:bpt-lemma}
    Let $\cC$ be a stabilizer code with connectivity graph $G = (V,E)$. Let $A_1, A_2, A_3 \subset V$ such that $A_1 \sqcup A_2 \sqcup A_3 = V$,
    with $A_1$ and $A_2$ correctable, then the number of logical qubits $k$ satisfies
    \[
    k \leq |A_3|
    \]
\end{lemma}

By definition, a graph of Nagata dimension one can be covered by two collections of subsets $\fB_1, \fB_2$. Each of the subsets in $\fB_1$ have constant size, and are far apart from each other; and similarly for $\fB_2$. Applying the Union Lemma, the region covered by $\fB_1$ is correctable if the distance is sufficiently large, and so it the region covered by $\fB_2$. As the whole graph can be decomposed into two correctable regions, Lemma \ref{lem:bpt-lemma} state that $k = 0$. We formalise this argument in Lemma \ref{lem:nagata-distance-bound}.

\begin{lemma}[$m\leq 1$ Nagata dimension implies constant distance]
\label{lem:nagata-distance-bound}
    Let $\cC$ be a stabilizer code with connectivity graph $G$ and $k\geq1$. If $G$ has Nagata dimension at most $1$ at scale $r\geq 2$ for some constant $c$, then the distance $d$ of $\cC$ obeys
    \[
    d \leq b_{\max}(2c)
    \]
    Where $b_{\max}(2c)$ is the maximum size of a ball of radius $2c$ in $G$.
\end{lemma}
\begin{proof}
    By assumption, there exist two collections of subsets $\fB_1, \fB_2$ such that
    \begin{enumerate}
        \item $\bigcup_{i=1}^{i=2} \bigcup_{B \in \fB_i} = V$
        \item $\fB_1$ and $\fB_2$ are each 2-separated
        \item Every subset in the collections $\fB_1, \fB_2$ is $2c$-bounded
    \end{enumerate}

    We assume, for the sake of contradiction, that $d>b_{\max}(2c)$. Then every $B \in \fB_1$ obeys $|B| \leq b_{\max}(2c) < d$, and we conclude that every $B$ is correctable. Further, since $\fB_1$ is $2$-separated, then for any $B,B' \in \fB_1$, we have $\partial B \cap B' = \emptyset$. Invoking the Union Lemma \ref{lem:union-lemma}, we obtain that $\bigcup_{B \in \fB_1} $ corresponds to a set of correctable qubits. The same argument applies to $\fB_2$.

    We can now write $A_1 = \bigcup_{B \in \fB_1} $, and $A_2 = \bigcup_{B \in \fB_2} \setminus A_1 $. Note that both $A_1$ and $A_2$ are correctable: removing qubits from a correctable set leaves a correctable set. As $(V \setminus A_1 ) \setminus A_2 = \emptyset$, we pick $A_3 = \emptyset$. The BPT Lemma \ref{lem:bpt-lemma} then gives $k \leq |A_3| = 0$: the code $\cC$ does not encode any logical qubits.
\end{proof}

We are now in position to state our main result. As the Nagata dimension by itself is not easy to compute in general, we use the fact that the Assouad dimension upper bound the Nagata dimension (Theorem \ref{thm:assouad-to-nagata}) to obtain an intuitive result. 

\begin{theorem}
\label{thm:main-thm}
    Let $\cC$ be a stabilizer code with connectivity graph $G$ and $k \geq 1$. If $G$ has Assouad dimension at most $\beta<2$ at scale $R\geq R_0$ with constant $C$, for some $R_0$ depending only on $\beta$ and $C$; then the distance $d$ of $\cC$ obeys

    \[
    d \leq C(4c)^\beta
    \]
    
\end{theorem}
\begin{proof}
    Applying Theorem \ref{thm:assouad-to-nagata} with $r_* = 2, \beta_* = \beta, C_* = C$ guarantees that $G$ has Nagata dimension at most $1$ at scale $2$ with some constant $c \equiv c_*$. Invoking Lemma \ref{lem:nagata-distance-bound}, we obtain that $d \leq b_{\max}(2c)$. Since the Assouad dimension is bounded, we can also bound $b_{\max}$. Picking $r=1/2$, note that balls of radius $1/2$ in a graph correspond to a single vertex. Definition \ref{def:assouad-dim} then guarantees that $b_{\max}$ is at most $C(\frac{2c}{1/2})^\beta = C(4c)^\beta$, hence the stated result.
\end{proof}
\section{Open Questions}
Quantum codes seem to obey a very distinctive phenomenon of `dimensional threshold'. The only codes known to be self-correcting live in $\geq 4D$ \cite{alicki2008thermal}, the only known codes to have a macroscopic energy barrier live in $>2D$ \cite{bravyi2009no}, and we have now proved that the only codes with growing distance live in $\geq 2D$.
 It would be interesting to investigate whether the notions of Assouad/Nagata dimensions can provide answers as to the current lack of self-correcting codes in less than 4D. Do codes in $<4D$ necessarily contain string operators?

\section*{Acknowledgements}

N.\,B. is supported by the Australian Research Council via the Centre of Excellence in Engineered Quantum Systems (EQUS) project number CE170100009, and by the Sydney Quantum Academy.

\newpage

\printbibliography


\end{document}